\documentclass[11pt]{article}

\usepackage[OT1]{fontenc}

\usepackage[english]{babel}

\usepackage[utf8]{inputenc}

\usepackage[sc]{mathpazo}

\usepackage{amsmath,amssymb,amsfonts,mathrsfs}

\usepackage[amsmath,thmmarks]{ntheorem}

\usepackage{graphicx}

\usepackage{varioref}

\usepackage{datetime}

\usepackage{mathtools}

\usepackage[h]{esvect}

\usepackage{array}

\usepackage{booktabs}

\usepackage{import}

\theoremseparator{.}

\newtheorem{theorem}{Theorem}

\newtheorem{remark}[theorem]{Remark}

\newtheorem{lemma}[theorem]{Lemma}
\newtheorem{proposition}[theorem]{Proposition}

\theoremstyle{nonumberplain}
\theorembodyfont{\normalfont}
\newtheorem{assumption}{Assumption}
\newtheorem{intuitive case}{Intuitive Case}

\theoremstyle{nonumberplain}
\theorembodyfont{\normalfont}
\theoremsymbol{\ensuremath{\square}}
\newtheorem{proof}{Proof}
\renewcommand{\epsilon}{\ensuremath\varepsilon}

\renewcommand{\phi}{\ensuremath{\varphi}}

\DeclareMathOperator{\Tr}{Tr}

\newcommand{\com}[1]{}

\newcommand*\circled[1]{\tikz[baseline=(char.base)]{
			\node[shape=circle, draw,inner sep=1pt] (char) {#1}; }}
			
\newcommand\myequ[1]{\mathrel{\overset{\makebox[0pt]{\mbox{\normalfont\tiny\sffamily #1}}}{=}}}

\usepackage[linkcolor=black,colorlinks=true,citecolor=black,filecolor=black]{hyperref}

\usepackage{tikz}

\title{Stahl's Theorem (aka BMV Conjecture): Insights and Intuition on its Proof}
\author{Fabien Clivaz \\Institute for Theoretical Physics, \\ETH Z\"urich, 8093 Z\"urich, Switzerland}
\date{Monday \(29^{\text{th}}\) February, 2016}

\begin{document}

\maketitle

\begin{abstract}
 The Bessis–Moussa–Villani conjecture states that the trace of \(\exp(A - tB)\) is, as a function of the real variable \(t\), the Laplace transform of a positive measure, where \(A\) and \(B\) are respectively a hermitian and positive semi-definite matrix. The long standing conjecture was recently proved by Stahl and streamlined by Eremenko. We report on a more concise yet self-contained version of the proof.
\end{abstract}

\section{Statement \label{sec-statement}}

In 1975, Bessis Moussa and Villani conjectured  in \cite{Bessis-75} a way of rewriting the partition function of a broad class of statistical systems. The precise statement can be formulated as follows, see \cite{Lieb-12,Lieb-04} for a popular reformulation by E. Lieb and R. Seiringer.

\begin{theorem}[Stahl's Theorem]
\label{thm-BMV}
Let A and B be two \(n \times n\) Hermitian matrices, where B is positive semidefinite. Then the function
\begin{equation}
f(t) := Tr \; e^{A-tB}, \; \; t \geq 0
\end{equation}
can be represented as the Laplace transform of a non-negative measure \(\mu\). That is,
\begin{equation}
f(t) = \int_0^{\infty} e^{-ts} d\mu(s).
\end{equation} 
\end{theorem}

More than 30 years later, after having raised the interest of many scientists \cite{Lieb-04, Moussa-00, Hillar-05, Johnson-05, Hansen-06, Hagele-07, Hillar-07, Klep-08, Fleischhack-10, Burgdorf-08}, Stahl published a proof of this conjecture in \cite{Stahl-12}. A minimal version of the proof has meanwhile been published by Eremenko in \cite{Eremenko-13}. Our aim is to reconcile the exactness of Stahl's version of the proof with the clarity of Eremenko's version.

\begin{intuitive case}
To get a feeling of why the above theorem holds, let us investigate the case where \(A\) and \(B\) commute.

Since our matrices are simultaneously diagonalisable, we can w.l.o.g. assume that they are given in diagonal form and exponentiating them becomes trivial. We therefore have:
\begin{equation}
f(t)=Tr \; e^{A-tB}= \sum_{j=1}^n e^{a_j} e^{-tb_j}.
\end{equation}

We next define the measure \(\mu:= \sum_{j=1}^n e^{a_j} \delta_{b_j}\), where \(a_j\) and \(b_j\) are the matrix elements of A and B, and \(\delta_{b_j}\) is the Dirac measure on \(\mathbb{R}\).
\com{ defined for \(\Omega \in \mathbb{R}\) by
\begin{equation}
\delta_{b_i} (\Omega)= \left\{
\begin{array}{l l}
0 & \quad \text{if } b_i \notin \Omega \\
1 & \quad \text{if } b_i \in \Omega
\end{array} \right.
\end{equation}

}By noting that for any function \(g (s)\), \( \int_0^{\infty} g(s) d\delta_{b_j} = g(b_j)\), one immediately sees that
\begin{equation}
 \int_0^{\infty} e^{-ts} d\mu(s) = \sum_{j=1}^n e^{a_j} e^{-t b_j} = f(t);
\end{equation}
showing that, in the case of commuting matrices, the BMV conjecture is realized with a discrete positive \(\mu\).
\end{intuitive case}

To simplify the analysis of the general case, we first prove the following

\begin{assumption}
W.l.o.g. \(B\) can be assumed to have distinct positive eigenvalues \(b_n > \cdots > b_1 > 0\).
\end{assumption}

\begin{proof}
Let \(B \geq 0\). We work in the diagonal basis of \(B\). We define \(B_{\epsilon}:=B+\epsilon \, D\) with \(D=diag(1, 2, \dots, n)\). Assuming Theorem \(\ref{thm-BMV}\) holds for \(B_{\epsilon}\), we want to prove it also holds for \(B\); that is, assuming \(\mu_{\epsilon}\) exists and is non-negative, we want to prove \(\mu\) exists and is non-negative.

Since the following involves the inverse Laplace transform, it is convenient to write the objects as tempered distributions. Explicitly,
\begin{equation}
\begin{aligned}
\mu_{\epsilon}[\phi] &:= \int_0^{\infty} \phi(s) d\mu_{\epsilon}(s),\\
f_{\epsilon}[\phi] &:= \int_0^{\infty} \phi(s) f_{\epsilon}(s) ds;\\
\end{aligned}
\end{equation}
for test functions \(\phi \in C_0^{\infty}(\mathcal{R^+})\). We note that \(\mu_{\epsilon} \geq 0 \Leftrightarrow \mu_{\epsilon}[\phi] \geq 0 ,\, \forall \phi \geq 0\).
Denoting the Laplace transform by \(\mathcal{L}\), we have:
\begin{equation}
\mathcal{L}(\mu_{\epsilon})[\phi]:=\mu_{\epsilon}[\mathcal{L}(\phi)]=f_{\epsilon}[\phi],
\end{equation}
which using the Bromwich integral formula yields 
\begin{equation}
\mu_{\epsilon}[\phi]=\mathcal{L}^{-1}(f_{\epsilon})[\phi]=\frac{1}{2\pi i} \int_{x-i \infty}^{x+i \infty} f_{\epsilon}(z) \left( \int_0^{\infty} e^{z(s)} \phi(s) ds \right) dz.
\label{equ-mu}
\end{equation}

Using the Dominated Convergence Theorem one shows that (see appendix A of \cite{Clivaz-14})
\begin{equation}
\lim_{\epsilon \rightarrow 0} \mu_{\epsilon}[\phi]=\frac{1}{2\pi i} \int_{x-i \infty}^{x+i \infty} f(z) \left( \int_0^{\infty} e^{z(s)} \phi(s) ds \right) dz \geq 0\;, \; \forall \phi \geq 0,
\end{equation}
where the inequality comes from \(\mu_{\epsilon} [\phi] \geq 0\;, \; \forall \phi \geq 0\). So with
\begin{equation}
\mu[\phi] := \frac{1}{2 \pi i} \int_{x - i \infty}^{x+ i \infty} f(z) \left(\int_0^{\infty} e^{z(s)} \phi(s) ds \right) dz,
\end{equation}
we have \(f= \mathcal{L}(\mu\)) and \(\mu \geq 0\).
\end{proof}

\section{Eigenvalues of \(A-tB\)}

We tackle the general case by looking at \(\lambda_1(t), \dots, \lambda_n(t) \); the eigenvalues of \(A-tB\) .
\begin{theorem} \hfill
\label{thm-lambda}
\begin{description}
\item[i)] \(\lambda_1,\dots,\lambda_n\) have no branch point over \(\mathbb{R}\).
\item[ii)] \(\lambda_1, \dots, \lambda_n\) are analytic in a neighborhood of infinity and \(\forall j=1,\dots,n \)
\begin{equation}
\lambda_j(t) = a_{jj} - t b_j + \mathcal{O}(\frac{1}{t}) \; (t \rightarrow \infty ).
\end{equation}
\end{description}
\end{theorem}

\begin{proof}
We want to study 
\begin{equation}
\begin{alignedat}{2}
&\det \left( \lambda(t) \, id - (A-tB) \right) &= 0 \; \text{as } t \rightarrow \infty \\
\Leftrightarrow &\det \left( b(u) \, id - (B+ u A) \right) &= 0 \; \text{as } u \rightarrow 0,
\end{alignedat}
\end{equation}
with 
\begin{equation}
\label{def-bu}
u:= -\frac{1}{t} \;\; \text{and} \; \; b(u) := u\cdot \lambda\left(-\frac{1}{u}\right).
\end{equation}

That is, we are interested in the form of \(b(u)\), the slightly perturbed (isolated) eigenvalues of \(B\). Fortunately, this finds an answer in most text books on Quantum Mechanics. See for e.g. ch. 11.1 of \cite{Schwabl-02} for an intuitive approach or ch. XII of \cite{Reed-78} for a rigorous one. In any case, one finds for \(j=1,\dots,n\):
\begin{equation}
\label{equ-bu}
b_j(u)=b_j + u a_{jj} + \mathcal{O}(u^2) \; (u \rightarrow 0).
\end{equation}

Analyticity and uniqueness of \(b_j(u)\) near \(u=0\) is assured by Theorem XII.1 in \cite{Reed-78} and since \(B+uA\) is self adjoint \(\forall u \in \mathbb{R}\), by Rellich's Theorem (Theorem XII.3 in \cite{Reed-78}), \(b_i(u)\) is analytic and single valued in a neighborhood of \(u_0\), \( \forall u_0 \in \mathbb{R}\). Plugging definition \ref{def-bu} in equation \ref{equ-bu} we therefore have for \(j=1,\dots,n\) that
\begin{equation}
\lambda_j(t)=a_{jj}-t b_j +\mathcal{O} \left(\frac{1}{t}\right) (t \rightarrow \infty)
\end{equation}
is analytic in a neighborhood of infinity and has no branch point over \(\mathbb{R}\).
\end{proof}

\section{Explicit Form of \(\mu\)}
We now postulate an explicit form for \(\mu\).
\begin{theorem}
\label{thm-mupost}
The measure \(\displaystyle \mu :=\omega+ \sum_{j=i}^n e^{a_{jj}} \delta_{b_j} \) satisfies
\begin{equation}
f(t)= \int_0^{\infty} e^{-ts} d\mu(s),
\end{equation}
for \(f(t)= \Tr{e^{A-tB}} \) and \(\displaystyle d\omega (s):= \omega(s) ds\), where 
\begin{equation}
\quad \omega(s):= \frac{1}{2 \pi i} \sum_{j:\, b_j< s} \int_{\partial U} e^{\lambda_j(z)+s z} dz;
\end{equation}
with \(U\) a neighborhood of infinity such that \(\partial U\) is a positively oriented Jordan curve around zero.
\end{theorem}

Before verifying Theorem \ref{thm-mupost}, we prove the useful
\begin{lemma} 
\label{lemma-suppw}
\(\text{supp}(\omega) \subset [b_1,b_n]\).
\end{lemma}

\begin{proof}
For \(s \leq b_1\) the sum \(\sum_{j: b_j < s}\) is void and hence trivially \(\omega(s)=0\).\\
For \(s > b_n\) we have:
\begin{equation}
\begin{aligned}
2\pi i \, \omega(s)&= \sum_{j=1}^{n} \int_{\partial U} e^{\lambda_j(z)+s z} dz\\
&=\int_{\partial U} \Tr{e^{A- z B}} e^{s z} dz,
\end{aligned}
\end{equation}
where we used the spectral decomposition definition of  \(e^{A-z B}\), that is 
\begin{equation}
e^{A-z B}:= \sum_{\lambda} e^{\lambda} P_{\lambda}; \quad \lambda \text{: Eigenvalue of } A- z B.
\end{equation}

Equivalently, see e.g. \cite{Kato-95}, one can define \(e^{A-z B}\) through
\begin{equation}
\label{def-opint}
e^{A-z B}:= \frac{1}{2 \pi} \int_{\gamma} \left( z' \, id - (A-z B)\right)^{-1} e^{z'} dz',
\end{equation}
with \(\gamma\) enclosing the spectrum of \(A-z B\), thereby ensuring \(\left( z' \, id - (A-z B)\right)^{-1}\) to be well defined for \(z' \in \gamma\) and in fact analytic as a function of \(z\), since for any fixed \(z' \in \gamma\) we have that
\begin{equation}
\begin{aligned}
\frac{d}{dz}\left( z' \, id - (A-z B)\right)^{-1} = &-\left( z' \, id - (A-z B)\right)^{-1} \left(\frac{d}{dz} \left( z' \, id - (A-z B)\right)\right) \\
&\left( z' \, id - (A-z B)\right)^{-1}.
\end{aligned}
\end{equation}

With definition \ref{def-opint} we therefore see that \(\Tr{e^{A-z B}}\) is analytic and hence by Cauchy's Theorem \(\omega(s)=0\).
\end{proof}

\begin{proof}[of Theorem \ref{thm-mupost}]
We want to verify that \(\mathcal{L}(\mu)=f\).

The first part of \(\mu\) is the expression we found in the intuitive case of section \ref{sec-statement}. Using Lemma \ref{lemma-suppw} and noting that  \(\sum_{j: b_j < s} = \sum_{j=1}^k\) for \(s \in (b_k, b_{k+1}]\), we find for the second one
\begin{equation}
\mathcal{L} (\omega) (t)= \int_{b_1}^{b_n} e^{-ts} \omega(s) ds = \sum_{k=1}^{n-1} I_k(t),
\end{equation}
with
\begin{equation}
I_k= \frac{1}{2 \pi i} \int_{b_k}^{b_{k+1}} \left( \sum_{j=1}^k \int_{\partial U} e^{\lambda_j(z)+s(z-t)} dz \right) ds.
\end{equation}

Since according to Theorem \ref{thm-lambda} the \(\lambda_j\)'s have no branch point over \(\mathbb{R}\), by Cauchy's Theorem, we can, without altering the result of the integral, deform \(U\) to \(U_1\), with \(U_1\) as in figure \ref{fig-contour}.
Inverting the sums, i.e. \(\sum_{k=1}^{n-1} \sum_{j=1}^{k} = \sum_{j=1}^{n-1} \sum_{k=j}^{n-1}\), and performing the s-integral, we then get
\begin{equation}
\sum_{k=1}^{n-1} I_k  = \underbrace{\frac{1}{2 \pi i} \sum_{j=1}^{n}\int_{\partial U_1} e^{\lambda_j(z)} \frac{e^{b_n (z-t)}}{z-t} dz}_{\circled{\(\diamond\)}}-\underbrace{\frac{1}{2 \pi i} \sum_{j=1}^{n} \int_{\partial U_1} e^{\lambda_j(z)} \frac{e^{b_j (z-t)}}{z-t} dz}_{\circled{\(\star\)}}.
\end{equation}
Note that the \(n^{th}\) summand of \(\circled{\(\diamond\)}\) and \(\circled{\(\star\)}\) cancel each other.

Since \(f(z) \frac{e^{b_n (z-t)}}{z-t}\) is entire in \((U_1)^c\), we have by Cauchy's Theorem
\begin{equation}
\circled{\(\diamond\)} = \frac{1}{2 \pi i} \int_{\partial U_1} f(z) \frac{e^{b_n (z-t)}}{z-t} dz =0.
\end{equation}

To evaluate \(\circled{\(\star\)}\) we first split the integration path:
\begin{equation}
\circled{\(\star\)}=\sum_{j=1}^n \left[ \underbrace{\frac{1}{2 \pi i} \int_{\partial U_1-C}e^{\lambda_j(z)} \frac{e^{b_j (z-t)}}{z-t} dz}_{\circled{1}} + \underbrace{\frac{1}{2 \pi i} \int_{C}e^{\lambda_j(z)} \frac{e^{b_j (z-t)}}{z-t} dz}_{\circled{2}}\right],
\end{equation}
with \(C\) a positively oriented curve with trace \(\{z : \; |z| = R >t\}\) as depicted in figure \ref{fig-contour}. Since \(z=t\) is the only pole enclosed by \(\partial U_1 - C\), using the residue Theorem, \(\circled{1}= - e^{\lambda_j (t)}\). We then rewrite \(\circled{2}\) using Theorem \ref{thm-lambda} to express \(\lambda_j\) as 
\begin{equation}
\lambda_j(z)= -b_j z + a_{jj}+r_j(z),
\end{equation}
where \(r_j\) is analytic in \( U_1\) and \(r_j(\infty)=0\). So 
\begin{equation}
\circled{2}=e^{a_{jj}-b_j t} \frac{1}{2\pi i}\int_C \frac{e^{r_j(z)}}{z-t} dz.
\end{equation}

\begin{figure}
\centering
\def\svgwidth{ 0.75 \columnwidth}
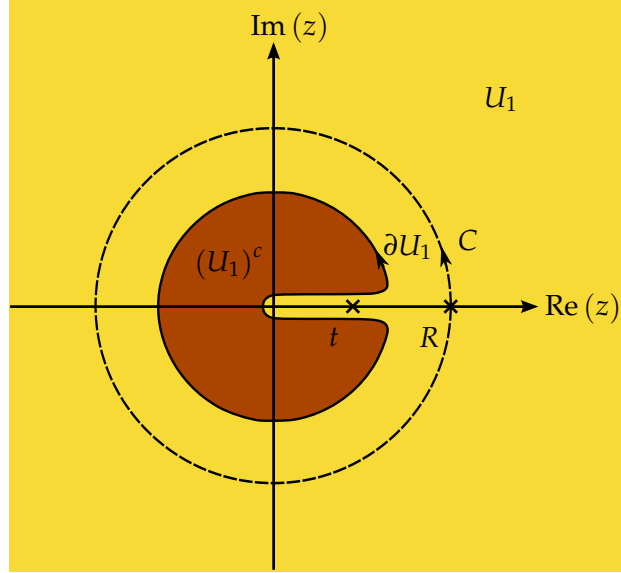
\caption{The choice of \(U_1\), in yellow, is made such that \(t \in U_1\). Such a choice is enabled by Theorem \ref{thm-lambda} i).}
\label{fig-contour}
\end{figure}

Performing the change of variable \(z:= 1/z\), the new variable integrates over
\begin{equation}
\frac{1}{C} : \text{negatively oriented curve with trace } \{z : \; |z|= \frac{1}{R}\},
\end{equation}
and we therefore get
\begin{equation}
\label{equ-circle2}
\circled{2}=- e^{a_{jj}-b_j t}\frac{1}{2 \pi i} \int_{\frac{1}{C}} \frac{e^{r_j(z^{-1})}}{z (1-tz)} dz= e^{a_{jj}-b_j t}e^{r_k(\infty)}=e^{a_{jj}-b_j t};
\end{equation}
since as \(|z|= 1/R <1/t\), the only pole of the integrand is at \(z=0\).

Gathering the results of \( \circled{1}, \; \circled{2}, \; \circled{\(\star\)} \text{ and } \circled{\(\diamond\)}\) we get
\begin{equation}
\mathcal{L} (\omega) (t)= -\circled{\(\star\)}= - \sum_{j=1}^{n} e^{a_{jj}-b_j t} + \Tr{e^{A-tB}},
\end{equation}
which with the result of the intuitive case gives \(\mathcal{L}(\mu)(t)= \Tr{e^{A-tB}}\).
\end{proof}

\section{Domain of Definition of \(\lambda\)}

We would now like to talk about \(\lambda\), the solution of \(det \left( \lambda(t) \, id - (A-tB) \right)=0\), in a global fashion instead of viewing it as \(n\) different functions \(\lambda_1, \dots, \lambda_n\). A fruitful way to do so is to define its domain of definition, \(S\), as a Riemann surface; for further reading see \cite{Farkas-92} or \cite{Teleman-03}.

We choose the \(n\) sheets of \(S\), \(S_j \; (j=1,\dots,n)\), such that in the neighborhood of infinity where we already numbered the \(\lambda_j\)'s (see Theorem \ref{thm-lambda}) we have that
\begin{equation}
\lambda_j=\lambda \circ \pi_j^{-1},
\end{equation}
with \(\pi: S \rightarrow \bar{\mathbb{C}}\) the canonical projection of \(S\) and \(\pi_j\) its restriction to \(S_j\).

We further denote the lifting of the complex conjugate over \(S\) by \(\rho\) and note that since \(\lambda\) is of real type, \(\rho (S_+)= S_-\) and vice versa; where
\begin{equation}
\begin{aligned}
S_+ &:= \{ \xi \in S \mid  \operatorname{Im} \pi(\xi)>0 \},\\
S_- &:=  \{ \xi \in S \mid  \operatorname{Im} \pi(\xi)<0 \}.
\end{aligned}
\end{equation}
That is, \(S\) is anti-conformal.
\section{Non-Negativity of \(\mu\)}

To conclude the proof of Theorem \ref{thm-BMV} we have to prove that \(\mu= \sum_{j=1}^n e^{b_j} \delta_{b_j} + \omega \geq 0\).
The first summand is obviously non-negative. To prove the second one is also non-negative, we need to show that
\begin{equation}
\omega(s)= \frac{1}{2\pi} \sum_{j : b_j <s} \int_{\partial U} e^{\lambda_j(z) + s z} dz \geq 0 ;\quad \forall s \in (b_1, b_n].
\end{equation}

To do so, we will replace the lift of \(\sum_{j: b_j<s} \int_{\partial U}\) on \(S\) by \(\int_{\gamma}\) on which the projection of the integrand is real and positive, for some well chosen contour \(\gamma\) on \(S\).

In the following we fix \(s \in (b_k, b_{k+1})\), with \(k \in \{1, \dots, n-1\}\) also fixed. The case \(s=b_{k+1}\) is achieved by continuity. We also write \(g:= \lambda + s \pi\).

\subsection{Constructing \(\gamma\)}

On \(S\) we define \(D:= \{ \xi \mid \frac{\operatorname{Im}g(\xi)}{\operatorname{Im}\pi(\xi)} >0 \}\). For \(\xi_0 \in \pi^{-1}(\mathbb{R})\) we note that since the \(\lambda_j\)'s have no branch point over \(\mathbb{R}\), we locally stay on the same sheet such that \(\pi\) locally has an inverse \(\pi^{-1}\). Thus, although \(\operatorname{Im}\pi(\xi_0)=0\), we can define the quotient as
\begin{equation}
\frac{\operatorname{Im}g(\xi_0)}{\operatorname{Im}\pi(\xi_0)} := \lim_{y \rightarrow 0} \frac{\operatorname{Im}g \circ \pi^{-1} (x_0,y)}{y},
\end{equation}
with \(\pi(\xi)=x+iy\equiv (x,y)\) and \(\pi(\xi_0)=(x_0,0)\). Furthermore, since \(\lambda \circ \pi^{-1}\) is of real type, \(\operatorname{Re} g\circ \pi^{-1} (x,y) \) is even in \(y\) and hence \( (\partial_2 \operatorname{Re} g \circ \pi^{-1}) (x_0,0)=0\), such that with l'H\^opital's rule we get
\begin{equation}
\label{equ-quotientD}
\frac{\operatorname{Im} g(\xi_0)}{\operatorname{Im} \pi(\xi_0)}= (\partial_2 g \circ \pi^{-1})(x_0,0);
\end{equation}
showing that the quotient is well defined for any \(\xi \in S\).
To help visualize \(D\), we note that \(\rho(D)=D\). A possible realization of \(D\) is depicted in figure \ref{fig-D}. We next look at \(\partial D=\{ \xi \mid \frac{\operatorname{Im}g(\xi)}{\operatorname{Im}\pi(\xi)} =0 \}\) and find using equation \ref{equ-quotientD} that \(\partial D \cap \pi^{-1}(\mathbb{R})\) is made of discrete points being in fact the continuation of the curves of \(\partial D \cap \pi^{-1}(\mathbb{R}^c) = (\operatorname{Im}g)^{-1}(\{0\}) \cap \pi^{-1}(\mathbb{R}^c)\) (see appendix B of \cite{Clivaz-14}). We propose \(\partial D\) to be the trace of \(\gamma\).

That \(\gamma\) is suited to prove the positivity of \(\mu\) is the content of
\begin{proposition}[The Crucial Link]
\label{prop-crucial}
\begin{equation}
 \frac{1}{2\pi i} \sum_{j: b_j < s} \int_{\partial U} e^{\lambda_j(z)+sz} dz = - \frac{1}{2\pi i} \int_{\gamma} e^{g(\xi)} d\xi >0.
\end{equation}
\end{proposition}

Indeed, proving it concludes the proof of Theorem \ref{thm-BMV}. Before doing so, we though look into some properties of \(\gamma\).

\begin{figure}
\centering
\def\svgwidth{ 0.75 \columnwidth}
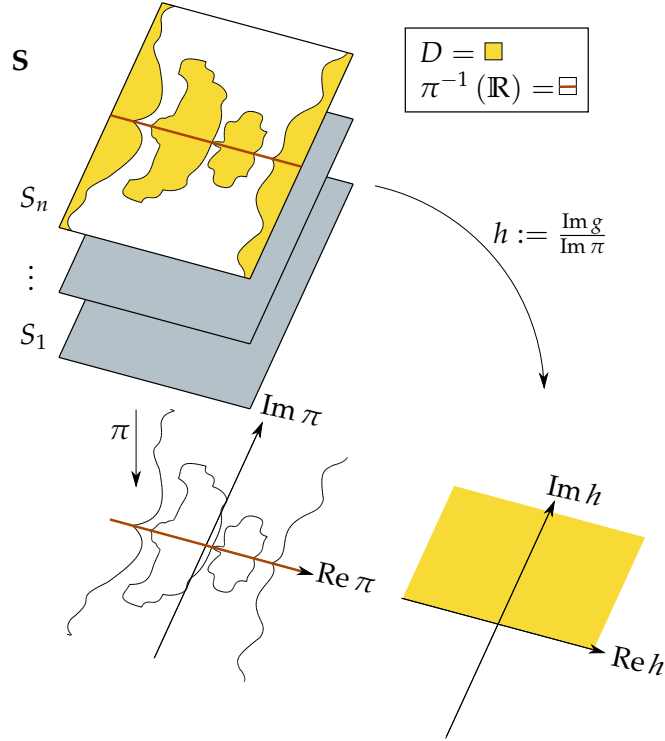
\caption{A possible representation of D on S as well as its image through h are displayed in yellow. Note the symmetry of D with respect to \(\pi^{-1}\left(\mathbb{R}\right)\), depicted in brown.}
\label{fig-D}
\end{figure}

\subsection{Properties of \(\gamma\)\label{sec-propgamma}}

\begin{lemma}
\label{lemma-gammaprop}
{\normalfont \textbf{ i)  }} \(\gamma = \gamma_1 + \dots + \gamma_N, \quad \gamma_i: \text{positively oriented Jordan curve.}\)
\begin{equation}
\begin{aligned}
{\normalfont \textbf{ii) }}& \operatorname{Re}(g\circ \gamma_k)& \text{monotonically increasing on }  \gamma_k^{-1}(S_+),\\
& & \text{monotonically decreasing on } \gamma_k^{-1}(S_-).
\end{aligned}
\end{equation}
\end{lemma}

\begin{proof}
From the above discussion, up to discrete points, the trace of \(\gamma\) is \((\operatorname{Im}g)^{-1}({0})\backslash \pi^{-1}(\mathbb{R})\). Since \(\operatorname{Im}g\) is a harmonic function, everywhere except at a finite number of critical points denoted by \(C_r\), \(\partial D\) is locally the trace of a unique curve (see appendix C of \cite{Clivaz-14} for a proof) . Furthermore, since \(\operatorname{Im}g\) is non-constant, any point of \(C_r\) is found to be a zero of order \(m<\infty\) of \(g\); and hence by the Auxiliary theorem of section 4.1 in \cite{Korevaar-11}, \(\partial D\) is the trace of exactly \(m\) curves around such points. Because of the anti-conformal structure of \(S\), those traces form closed loops; allowing us to choose \(\gamma\) as in i).

As depicted in figure \ref{fig-gamma}, \(\operatorname{Im}g\) changes sign when one crosses the trace of \(\gamma_k\). Choosing the axis \(l\) along \(\gamma_k\) and \(n\) normal to it pointing in \(D\), this means
\begin{equation}
\begin{aligned}
\frac{\partial}{\partial n} \operatorname{Im}g(\xi) > 0&  \quad \forall \xi \in \gamma \cap S_+,\\
\frac{\partial}{\partial n} \operatorname{Im}g(\xi) < 0& \quad \forall \xi \in \gamma \cap S_-,
\end{aligned}
\end{equation}
\begin{figure}
\centering
\def\svgwidth{ 0.75 \columnwidth}
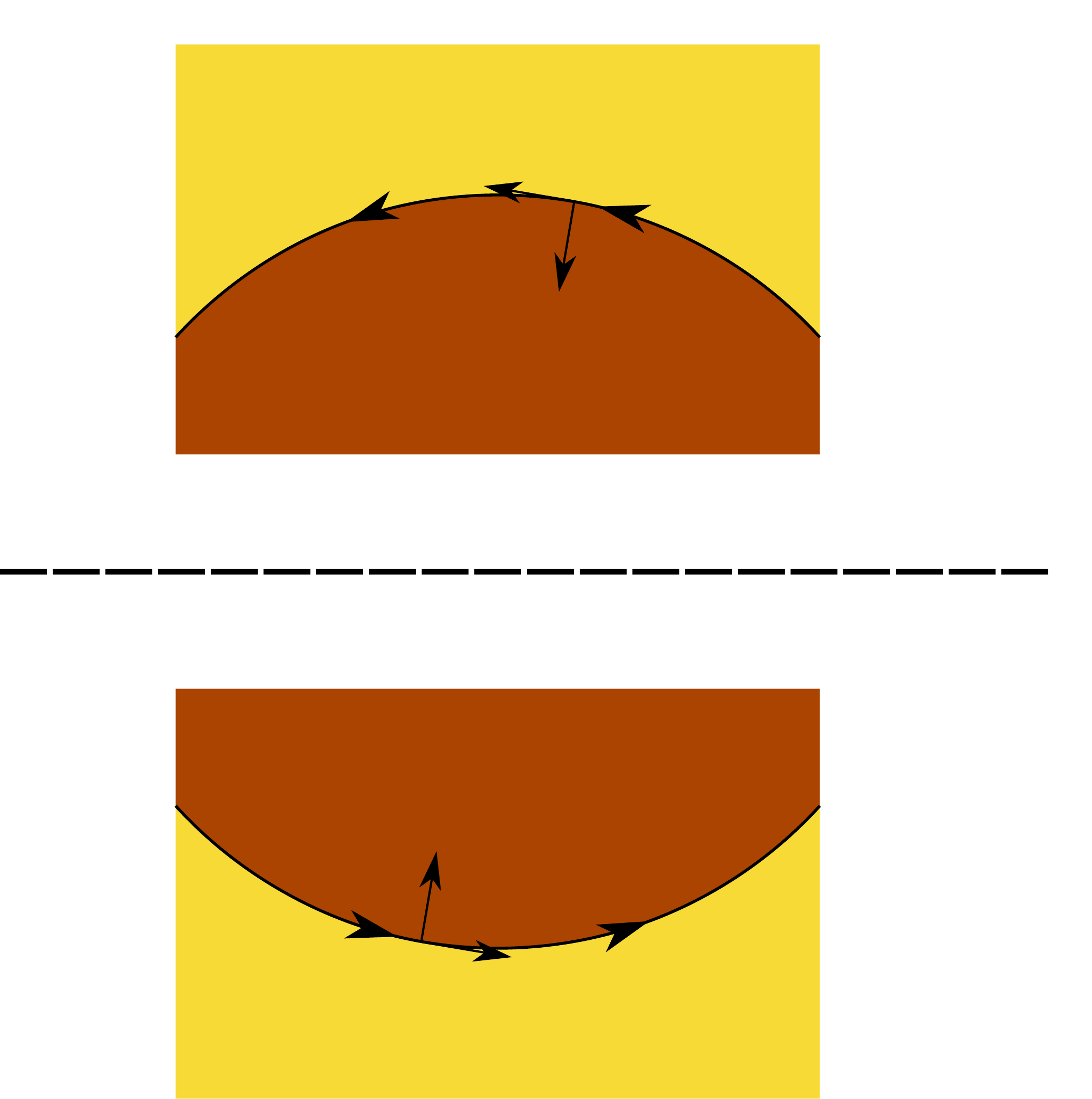
\caption{ The curve \(\gamma_k\) is locally depicted in a region of \(S_+\), top, and \(S_-\), bottom. The change of sign of \(\operatorname{Im}g\) is observed when crossing \(\gamma_k\) along the local coordinate \(\protect \overrightarrow{n}\).}
\label{fig-gamma}
\end{figure}
which using the Cauchy-Riemann equations gives
\begin{equation}
\begin{aligned}
\frac{\partial}{\partial l} \operatorname{Re}g(\xi) > 0 & \quad \forall \xi \in \gamma \cap S_+,\\
\frac{\partial}{\partial l} \operatorname{Re}g(\xi) < 0 & \quad \forall \xi \in \gamma \cap S_-;\\
\end{aligned}
\end{equation}
proving ii).
\end{proof}

\begin{remark}
\label{remark-rho}
Since the \(\gamma_k\)'s are single valued, point ii) of Lemma \ref{lemma-gammaprop} tells us that each \(\gamma_k\) has to be contained in both \(S_+\) and \(S_-\), allowing us to chose \(\gamma\) such that the endpoints of \(\gamma_k\) lie on \(\pi^{-1}(\mathbb{R})\) and \(\rho(\gamma_k)= - \gamma_k\).
\end{remark}

\begin{lemma}
\label{lemma-intgamma}
\(\displaystyle -\frac{1}{2\pi i} \int_{\gamma_k} e^{g(\xi)} d\xi > 0.\)
\end{lemma}

\begin{proof}
As \(\gamma_k \subset(\operatorname{Im}g)^{-1}(\{0\})\), we have that 
\begin{equation}
\label{equation-img}
\operatorname{Im}g(\xi)=0 \quad \forall \xi \in \gamma_k;
\end{equation}
and hence together with Lemma \ref{lemma-gammaprop} ii)
\begin{equation}
\label{equation-expg}
e^{g \circ \gamma_k}= e^{\operatorname{Re} g \circ \gamma_k} \text{ is monotonically increasing on } \gamma_k^{-1} (S_+). 
\end{equation}

Writing \(z=x+iy\) and \(\pi^{-1}(z)= \xi= \nu + i \eta\) , we therefore get:
\begin{equation}
\begin{aligned}
\frac{1}{2\pi i} \int_{\gamma_k} e^{g(\xi)} d\xi &= \frac{1}{2\pi i} \left( \int_{\gamma_k \cap S_+} e^{g(\xi)} d\xi +\int_{\gamma_k \cap S_-} e^{g(\xi)} d\xi \right)\\
&\myequ{Rem.\ref{remark-rho}}  \frac{1}{2\pi i} \left( \int_{\gamma_k \cap S_+} e^{g(\xi)} (d\nu+i d\eta) +\int_{-\gamma_k \cap S_+} e^{\overline{g(\xi)}} (d\nu-i d\eta) \right)\\
&\myequ{equ.\ref{equation-img}}  \frac{1}{\pi}  \int_{\gamma_k \cap S_+} e^{g(\xi)} d\eta \\
&=  \frac{1}{\pi}  \int_{\gamma_k^{-1} (S_+)} e^{g\circ \gamma_k (s)} \operatorname{Im} \left( \pi_{\lambda}\circ \gamma_k \right)' ds \\
&\myequ{IBP}  -\frac{1}{\pi}  \int_{\gamma_k^{-1} (S_+)} \underbrace{\left(e^{g\circ \gamma_k (s)}\right)'}_{>0} \underbrace{\operatorname{Im} \left( \pi_{\lambda}\circ \gamma_k \right)}_{>0} ds <0,\\
\end{aligned}
\end{equation}
where the boundary terms when performing the integration by parts (IBP) vanish because of Remark \ref{remark-rho}.
\end{proof}

\subsection{Proof of The Crucial Link}
We finally prove Proposition \ref{prop-crucial}, which finishes the proof of Theorem \ref{thm-BMV}.

\begin{proof}[of Proposition \ref{prop-crucial}]
From Theorem \ref{thm-lambda}, we can choose a neighborhood of infinity, \(U\), such that \(\lambda\) has no branch point over \(\pi^{-1}\left(\overline{U}\right)\). That means that \(\pi^{-1} \left(\overline{U}\right)\) is made of \(n\) disjoint components, each fully in one sheet. We can hence write:
\begin{equation}
\pi^{-1} \left( \partial U\right)= C_1 \cup \dots \cup C_n\; ; \quad C_j \subset S_j, \quad \forall j =1,\dots, n.
\end{equation}

In \(\overline{U}\) we furthermore have by Theorem \ref{thm-lambda} ii) that
\begin{equation}
\operatorname{Im} \left(\lambda_j(z)+sz\right) = (s-b_j) \operatorname{Im}(z) + \operatorname{Im} \left( \mathcal{O} \left(\frac{1}{z}\right)\right).
\end{equation}

So since \(s \in (b_k, b_{k+1})\), for \(\lvert z \rvert > R\), \(R>0\) big enough we achieve:
\begin{equation}
\begin{aligned}
j \leq k: \quad& \operatorname{Im} \left( \lambda_j(z) + sz \right) \text{ has same sign as } \operatorname{Im} (z),\\
j > k: \quad& \operatorname{Im} \left( \lambda_j(z) + sz \right) \text{ has opposite sign as } \operatorname{Im} (z).\\
\end{aligned}
\end{equation}

Choosing \(\overline{U} \subset \{z \mid \lvert z \rvert > R\}\) we have :
\begin{equation}
\begin{alignedat}{2}
&C_1, \dots, C_k && \in D, \\
&C_{k+1}, \dots, C_n \quad &&\notin D.
\end{alignedat}
\end{equation}

Defining \(D_0:= D \setminus \pi^{-1}(\overline{U})\), we find with the above \(\partial D_0 = \gamma + C_1 + \dots + C_k\) and, since \(D_0\) is bounded, by Cauchy's Theorem 
\begin{equation}
\frac{1}{2 \pi i} \int_{\partial D_0} e^{g(\xi)} d\xi = 0;
\end{equation}
that is
\begin{equation}
\begin{aligned}
- \frac{1}{2 \pi i} \int_{\gamma} e^{g(\xi)} d\xi &= \frac{1}{2 \pi i} \sum_{j=1}^k \int_{C_j} e^{\lambda (\xi) + s \pi(\xi)} d\xi\\
&= \frac{1}{2 \pi i} \sum_{j=1}^k \int_{\partial U} e^{\lambda_j (z) + s z} dz\\
&= \frac{1}{2 \pi i} \sum_{j: b_j < s} \int_{\partial U} e^{\lambda_j (z) + s z} dz,\\
\end{aligned}
\end{equation}
which together with Lemma \ref{lemma-intgamma} proves the assertion.
\end{proof}

%we use unsrt instead of plain since we prefer the citation number to be sorted by order of appearance rather than alphabetically
\bibliographystyle{unsrt}

\bibliography{refs_prime}

\end{document}